%% file: main.tex
\documentclass[manuscript,11pt]{acmart}

\copyrightyear{2019}
\acmYear{2019} 
\setcopyright{iw3c2w3}
\acmConference[WWW '19]{Proceedings of the 2019 World Wide Web Conference}{May 13--17, 2019}{San Francisco, CA, USA}
\acmBooktitle{Proceedings of the 2019 World Wide Web Conference (WWW '19), May 13--17, 2019, San Francisco, CA, USA}
\acmPrice{}
\acmDOI{10.1145/3308558.3313679}
\acmISBN{978-1-4503-6674-8/19/05}

\usepackage{booktabs} 

\usepackage{amssymb}
\usepackage{amsmath,amsfonts,amsthm,amssymb}
\usepackage{enumerate}
\usepackage{setspace}
\usepackage{fancyhdr}
\usepackage{extramarks}
\usepackage{algorithm}
\usepackage{algorithmic}
\usepackage{chngpage}
\usepackage{soul,color}
\usepackage{comment}
\usepackage{graphicx,float,wrapfig}
\usepackage{subcaption}

\usepackage{fullpage}

\usepackage[utf8]{inputenc} 
\usepackage{url}            
\usepackage{booktabs}       
\usepackage{amsfonts}       
\usepackage{nicefrac}       
\usepackage{microtype}      

\newtheorem{theorem}{Theorem}[section]
\newtheorem{corollary}[theorem]{Corollary}

\newtheorem{proposition}[theorem]{Proposition}
\newtheorem{definition}[theorem]{Definition}
\newtheorem{lemma}[theorem]{Lemma}

\newtheorem{fact}{Fact}[section]
 
\newcommand{\E}{\mathbb{E}}
\newcommand{\N}{\mathcal{N}}

\newcommand{\Title}{Diversity and Exploration in Social Learning}
\newenvironment{prevproof}[2]{\noindent {\em {Proof of {#1}~\ref{#2}:}}}{$\hfill\qed$\vskip \belowdisplayskip}
\newcommand{\xhdr}[1]{\vspace{2mm} \noindent{\bf #1}}

\begin{document}
\title{\Title}

\author{Nicole Immorlica}
\affiliation{\institution{Microsoft Research}}
\email{nicimm@microsoft.com}

\author{Jieming Mao}
\affiliation{\institution{University of Pennsylvania}}
\email{maojm517@gmail.com}

\author{Christos Tzamos}
\affiliation{\institution{University of Wiscosin-Madison}}
\email{tzamos@wisc.edu}

\begin{abstract}
  
In consumer search, there is a set of items. An agent has a prior over her value for each item and can pay a cost to learn the instantiation of her value.  After exploring a subset of items, the agent chooses one and obtains a payoff equal to its value minus the search cost.  We consider a sequential model of consumer search in which agents' values are correlated and each agent updates her priors based on the exploration of past agents before performing her search.  Specifically, we assume the value is the sum of a common-value component, called the quality, and a subjective score.  Fixing the variance of the total value, we say a population is more diverse if the subjective score has a larger variance.  We ask how diversity impacts average utility.  We show that intermediate diversity levels yield significantly higher social utility than the extreme cases of no diversity (when agents under-explore) or full diversity (when agents are unable to learn from each other) and quantify how the impact of the diversity level changes depending on the time spent searching. 
\end{abstract}
%
  

\maketitle

\input{intro}

\input{related}

\input{model}

\input{mainproofs}

\input{diamond}


\appendix
%

\input{gaussian}

\input{omitted-proofs}

\newpage
\bibliographystyle{alpha}
\bibliography{bib}


\end{document}

%% file: intro.tex

\section{Introduction}

Members of a society often leverage the experiences of others to inform their own choices.  A high school student contemplating colleges might ask former high school students about their decision-making processes: what schools were considered, why the chosen one was selected over the others, and so on.  A tourist contemplating restaurants might read reviews of former diners on review websites: what was the overall impression of the restaurant, what aspects contributed to this impression, etc.  This sort of communal learning is highly valuable.  By listening to experiences from the past, society can avoid repeating the mistakes of the past.  However, this benefit comes at a potentially significant cost.  As each individual optimizes their own choice, they might naturally select a tried-and-true option over an unknown one.  As a result, unknown options, which could have great potential, are never discovered.

The classic field of {\em social learning}, initiated by ~\cite{Banerjee,Bikhchandrietal}, models this problem in the following way.  Agents arrive sequentially and choose an action (e.g., search strategy) that results in some payoff based on the state of the world (e.g., the values of the items).  Each agent receives a private signal correlated with the state of the world (e.g., priors on item values) and also observes chosen actions of past agents.  The literature derives conditions under which the society may ``herd'' on a sub-optimal action, ignoring private signals and instead exploiting information gleaned from the history of actions.  

The nascent field of {\em incentivizing exploration} tries to use information design or mechanism design to prevent the exploitation of historical information.  In this line of work, there is a principal (the review website in the case of restaurants for instance), a sequence of agents (tourists in the restaurant example), and a set of options (the restaurants themselves).  Each option has an inherent value, which is common to all agents, and the goal is to discover the highest-value option as quickly as possible.  The principal either explicitly pays agents to try new options~\cite{FrazierKKK14}, or sends signals that cleverly obfuscate the past experiences in a way that convinces agents to try new options~\cite{KremerMP13,MansourSS15}.  While these approaches successfully mitigate the dangers of learning from the past, they require the existence of a principal to coordinate society's search.  

In this paper, we study decentralized exploration.  Our central hypothesis is that societies perform sufficient search, even without a coordinating principal, when there are {\em sufficiently diverse preferences}.  Intuitively, this diversity should be 1) high enough that agents occasionally explore even when past agents had great experiences, and 2) low enough that there is still value in observing the past.  In this case, then society can perhaps learn to select good options, both by learning from each individual agent as she investigates her options before making a choice, and by learning from the sequence of agents who make different choices than their peers.  We ask, is there a simple model and large range of parameters in which we can observe this effect?

We explore our hypothesis in the context of a standard model of consumer search, iterated over time.  In a single round of consumer search, an agent is presented with a set of items.  The agent knows the distribution of her value for each item and may pay a fixed cost to learn its instantiation.  The resulting utility of the agent is her value for the item she selects minus the costs she paid to learn her values.  For example, a high school student contemplating colleges might know her approximate value for each and, based on this information, pick a subset to visit prior to selecting one.  In our iterated version of this model, what an agent knows about her values depends on what options previous agents considered.  More precisely, before performing her search, each agent updates her priors on the values observed by previous agents during their searches. In the college choice problem, this corresponds to a high school student listening to the trip reports of students in prior years before selecting which subset of schools to visit.  

The degree to which historical information impacts an agent's search depends on the amount of correlation, or {\em diversity}, in the agents' preferences.  In order to study the impact of diversity on average utility, we introduce a model in which each option has a common value component, drawn IID from a zero-mean Gaussian with variance $1-\sigma^2$, and each agent/option pair has a private value component, also drawn IID from a zero-mean Gaussian with variance $\sigma^2$.  Thus the total value of each option for each agent, the sum of the common and private value components, is also a zero-mean Gaussian with a fixed variance of $1$.  The variance parameter, $\sigma^2$, can be thought of as the diversity level.  For $\sigma^2=0$, all the value of an option comes from the common value component, whereas for $\sigma^2=1$, all the value comes from the private value component.  Since the distribution of option values remains constant as we change the diversity level, we are able to compare the average value obtained by societies of different diversity levels.

Our theoretical and simulation results show that diversity helps in this simple model.  In particular, intermediate levels of diversity always yield higher average utilities than the extremes of zero diversity ($\sigma^2=0$), in which preferences are perfectly correlated, or total diversity ($\sigma^2=1$), in which preferences are independent across agents.  We further investigate the comparison between intermediate levels of diversity. Our main result, Theorem~\ref{thm:main-intermediate}, Theorem~\ref{thm:main-large} and simulations, identify three distinct behaviors: For very short time horizons, less diverse societies outperform more diverse ones as the benefits of the added exploration have yet to take effect.  For very long time horizons, less diverse societies again outperform more diverse ones as even the less diverse societies eventually perform sufficient exploration.  But for a very large range of intermediate time horizons, we see more diverse societies outperform less diverse ones.  This is because the less diverse societies tend to exploit historical information at the expense of exploration, similar to the phenomenon of herding in the social learning literature. 

In addition to our aforementioned results, in Section~\ref{sec:diamond}, we extend our results to a more general setting that allows the quality of an item to get a very high value $D$ with tiny probability $\frac 1 D$. This setting models situations where even though the society might be very diverse and agents have almost no common value, in the unlikely event that some clearly superior option with value $D$ appears, they all agree that its value is high. We use this model to capture settings where the item values are not very concentrated, as Gaussians are, but have heavier
probability tails. These happen in many natural situations where breakthroughs with significantly high value occur and the whole society agrees that they are valuable. We show that for such settings, high diversity is extremely beneficial for the social good as it increases incentives for exploration and breakthroughs occur faster. The literature on incentivizing exploration~\cite{FrazierKKK14,HanKQ15} usually refers to these settings as ``Diamonds in the rough'' and often uses this setting to provide lower bounds.  In contrast, we show diverse societies excel in these settings.
\paragraph{Techniques} To establish our results, we need to be able to describe the optimal behavior of agents and bound the society's average expected utility. As we argue, the optimal behavior of the agents is characterized in \cite{Weitzman79}.  As shown in that paper, agents define an index for each item and then search them in decreasing order of index.  Unfortunately, this index is defined implicitly, making it difficult to explicitly determine the utilities obtained by the agents.  A crucial result that allows us to perform our analyses is that the average utility is approximately equal to the maximum quality of an item that agents will have explored by the end of the time horizon. This result enables us to focus on bounding the maximum quality of explored items. To provide an upper bound on this value, we consider an an alternative simpler scenario where agents, instead of following the optimal policy and computing indices of the inferred distributions of values, they use a plain threshold stopping rule. We show that there exists a coupling of the two processes such that the simpler process always explores more items and then proceed to bound the highest quality achievable when using the simple threshold stopping rule for every agent. We also give a matching lower bound by bounding the number of items that agents with a given diversity level could explore and then bound the expectation of the maximum value that these many normally distributed items may have.

We show that our bounds are asymptotically tight. This is required in order to be able to perform comparisons across different diversity levels as even a constant approximation of the expected average utility would not suffice to show a clear separation as average utilities of different diversity levels can be very close multiplicatively. To prove these sharp bounds, we rely on strong concentration results for Gaussian distributions.




%% file: related.tex

\section{Related Work}
Our work is very related to the line of work in Bayesian social learning.  The classic results of \cite{Banerjee,Bikhchandrietal} show that in social learning with homogeneous payoffs, agents may ignore their private signals in favor of information they infer from previous agents.  When this happens, society herds on sub-optimal actions.  A recent strand of literature, including \cite{Acemoglu11,Acemoglu17,Goereeetal,Lobel15a,Lobel15b}, explore the impact of heterogeneous preferences on social learning, generally observing that heterogeneous preferences help learning when the society is well-connected. The primary difference between these works and our paper is the structure of the actions and information sets.  With the exception of \cite{Goereeetal}, these works assume that the action of each agent is binary and the state of world is binary. On the contrary, in our paper, the action of each agent is a search strategy and the state of world (the qualities of items) is sampled from independent Gaussian distributions. This choice allows us to study the impact of diversity in a specific setting of interest, as motivated by the literature on consumer search initiated by \cite{Weitzman79}.  For the dispersion of information, \cite{Acemoglu11,Lobel15a,Lobel15b} assume agents observe a subset of past actions based on the social network and \cite{Acemoglu17,Goereeetal} assume agents observe the full history or summary statistics of the reviews of past agents where a review is a discrete valued function of an agent’s utility. In our paper, we assume agents observe actions of all previous agents and qualities/values of items explored by all previous agents.

Incentivizing exploration has been studied in the multi-armed bandit setting. In this line of research, it is assumed that myopic agents arrive sequentially and each pulls an arm at a time. As each agent only wants to maximize its own reward, exploration might not be done sufficiently without a principal coordinating exploration. \cite{KremerMP13,MansourSS15} study the setting where the outcomes of the actions are only observed by the principal and the principal convinces myopic agents to explore by sending signals containing partial information of the past outcomes. \cite{FrazierKKK14} considers a different setting where the myopic agents observe all the previous outcomes and the principal is allowed to use money transfers to incentivize exploration. A great amount of follow-up work has been done in various settings, including \cite{CheH15,HanKQ15,MansourSSW16,KannanKMPRVW17,Papanastasiou18}. 

There is an emerging line of work studying ``greedy algorithms" in multi-armed bandits\cite{bastani2017exploiting,kannan2018smoothed,externalities-colt18}. Greedy algorithms can be viewed as a form of social learning and these papers show that near-optimal performance can be achieved under certain diversity conditions.  

\cite{SR17} considers myopic agents with diverse preferences in the multi-armed bandit setting. They show that if subjective preferences are known to agents beforehand, myopic agents will explore all the arms even without a principal coordinating exploration. They consider a setting of qualities and subjective preferences such that the difference between any two qualities is smaller than 1 and the subjective preferences have random binary values. For each arm, this setting guarantees the existence of agents who prefer the arm no matter what quality is. Exploration is therefore guaranteed in the long run.

\cite{AnalytisSGM17} provides simulation results (without theoretical guarantees) showing that certain amount of preference diversity can lead to an increase in the average utility in a model that is closely related to our model. The main difference to our model is that they assume each agent searches items according to popularity levels of items and stops searching once an item has value larger than some threshold. On the contrary, we assume each agent performs the optimal search.

%% file: model.tex

\section{Model}
We consider a setting with infinitely many items\footnote{This assumption allows us to avoid edge cases in our analyses; our results carry over to settings with finitely many items so long as there are sufficiently many of them (in particular, more than the logarithm of the time horizon).} of unknown value. Agents arrive sequentially in $T$ rounds. The agent that arrives in round $t$ has a value $v_{i,t}$ for item $i$ that is the sum of two terms: an objective score $q_i$ for the item and an individual subjective score $s_{i,t}$.
The objective scores (qualities) $q_i$ are the same for all agents and are drawn from a Gaussian $\N(0,1-\sigma^2)$ with mean zero and variance $1-\sigma^2$ for some $\sigma\in[0,1]$. We refer to this score as the quality of item $i$.
The subjective scores $s_{i,t}$ vary among agents and are drawn independently from a Gaussian $\N(0,\sigma^2)$ with mean zero and variance $\sigma^2$. The total value of an item, $v_{i,t}=q_i+s_{i,t}$, is thus distributed according to a Gaussian $\N(0,1)$ with mean zero and variance one.   
\begin{definition}[Diversity Level]
	We define the {\em diversity} of a population of agents to be the parameter $\sigma$ (the standard deviation of the subjective score). 
\end{definition}
\noindent When the diversity is $0$, all the value of an item is derived from its objective score, i.e., $v_{i,t}=q_i$, and so agents' values are perfectly correlated.  When diversity is $1$, all the value of an item is derived from its subjective score, i.e., $v_{i,t}=s_{i,t}$, and so agents' values are independent.  Intermediate diversities interpolate between these two extremes.

Initially, the agent only knows a prior $\mathcal{F}_{i,t}$ over her value for each item $i$, which is a function of the initial Gaussian priors and her observations from past agents' searches. The agent may learn the objective score $q_i$ and subjective score $s_{i,t}$ of an item $i$, and hence its value $v_{i,t}=q_i+s_{i,t}$, by paying a search cost $c$, where $c$ is a fixed constant smaller than $\frac{1}{\sqrt{2\pi}}$ (chosen so that the agents are incentivized to explore at least one item). A search strategy of an agent is a mapping from observations to either a choice from among previously explored items, or an item to explore next.  The utility $u_{t}(\sigma, \{\mathcal{F}_{i,t}\})$ of the agent for a given search strategy and priors is her expected value for the chosen item $i^*$, $v_{i^*,t}$, minus her total expected search cost, i.e., $c$ times the expected number of items she explores. The chosen item $i^*$ has to be an explored item.

The optimal search strategy $\sigma^*$ for an agent in this setting was characterized by \cite{Weitzman79} in a slightly more general setting called the {\em Pandora's problem}.  This strategy assigns a score, called the {\em Weitzman index}, to every item and then explores items with positive indices in decreasing order of their indices until the observed value is greater than the following item's index. 

\begin{definition}[Weitzman Index]
	The {\em Weitzman index} $x_{i,t}^*$ is defined as
$\E_{v_{i,t} \sim \mathcal{F}_{i,t}}[\max(v_{i,t} - x_{i,t}^*,0)] = c$.
\end{definition}

\noindent We assume that all agents use the optimal strategy and explore items in decreasing order of Weitzman index.  We denote by $x^*$ the Weitzman index of an item that has not been explored previously, i.e. the solution to
$\E_{v \sim \N(0,1)}[\max(v - x^*,0)] = c$. The choice of $c < \frac{1}{\sqrt{2\pi}}$ guarantees that $x^*>0$ and thus every agent will explore at least one item.   

Once an item is explored (but not necessarily chosen) by some agent, agents in subsequent rounds obtain information about that item.  There are multiple choices we can make about exactly what information agents learn; the right choice is dictated by the setting and whether agents can infer each others' subjective biases.  In this paper, we will focus on a ``revealed quality'' model, where, after an item $i$ is explored by some agent, its quality $q_i$ is revealed to all agents in future rounds. In this model, agents update their prior for a previously explored item $i$ to $F_{i,t}=\N(q_i,\sigma^2)$ before computing its Weitzman index $x_{i,t}^*$.  We also consider a ``revealed value'' model, where after an item $i$ is explored by some agent in round $t$, only its value to that agent $v_{i,t} = q_i + s_{i,t}$ is revealed to future agents.  In this model, the optimal strategy of the agents becomes more complicated. Given the past values of agents for a given item, the current agent forms a posterior distribution of the item's quality. This posterior can be expressed as a Gaussian distribution with mean and variance that depend on the average past observed values and the number of times that item gets explored. We can show similar results in the ``revealed value'' model. We omit them because of the space limit.



Our goal is to understand how the expected average utility changes depending on the diversity level $\sigma$ of the society.  The impact of diversity will vary with the time spent searching, and so we define the average utility with respect to a time horizon $T$.  Let $\mathcal{H}_t$ be a history indicating the information obtained about the items explored in the first $t-1$ rounds and $\mathcal{F}(\mathcal{H}_t)$ be the induced priors for the agent arriving in round $t$.

\begin{definition}[Average Utility]
We define $A(\sigma, T)=\frac{1}{T}\sum_{t=1}^TE_{\mathcal{H}_t}[u_{t}(\sigma,\mathcal{F}(\mathcal{H}_t))]$ to be the {\em average expected utility} over agents in the first $T$ rounds when the diversity level is $\sigma$. 
\end{definition}


\subsection{Justification of Our Model}
Here we provide the rationale behind the model choices.
 
  \xhdr{Decomposition in Common and Private Values} The assumption that agent values decompose in independent common and private components  is commonly used in the empirical auction literature. Going back to the work of \cite{LiV98} who showed how to perform the decomposition based on deconvolution, it has been shown that it predicts the agent behavior quite accurately in real-world data \cite{Krasnokutskaya11, Asker10}. 
  
  \xhdr{Normally Distributed Values} The choice of normal distributions for the common and private values enables expressing a wide range of diversity levels with a single parameter. The constant sum of variances property allows for fair comparisons as in all cases agents have the same distribution of values for every option. Normal distributions appear in practice; e.g., the distribution of meta-critic scores for movies on Rotten Tomatoes closely fits a normal distribution; and the distribution of the common value in the US offshore oil and gas auctions has been shown to fit a normal distribution using econometrics~\cite{SyrgkanisTZ18}.  
  
  \xhdr{Pandora's Box Search Model} The Weitzman search strategy is the optimal policy that an agent can follow to maximize his expected utility. It has been successfully applied to real-world applications in the marketing literature  \cite{BartKM16,KimAB10,HonkaC17}.

%% file: mainproofs.tex

\section{How Different Diversity Levels Compare}

There are two opposing forces that affect average utility -- exploration and exploitation. When the society is not very diverse, agents tend to exploit previously explored options as it is quite costly for them to search for new alternatives that could potentially give them a higher value. On the other hand, diverse societies tend to explore more, but cannot exploit as much since information collected in the past is less valuable to them.

Focusing solely on exploration or exploitation can come at a significant cost.  In the extreme case where the society is not at all diverse (i.e., $\sigma=0$) and focuses on exploitation, we can show the society's average expected utility, $A(0,T)$, can be upper bounded by a constant. In particular, $A(0,T)$ does not depend on the number of rounds $T$.

\begin{lemma}
\label{lem:nodivub}
For any $T$, $A(0,T)\leq\sqrt{ \ln \left( \frac{1}{2\pi c^2}\right)}  + 2$. 
\end{lemma}

\noindent To prove Lemma~\ref{lem:nodivub}, we first bound by a constant the Weitzman index $x^*$ of an unsearched item, i.e.\ one whose quality is not yet revealed. We then argue that once some item with sufficient quality has been identified (with quality $q_i > x^*+c$), no other item will be explored in future rounds. We conclude the proof by upper-bounding the expected quality of such an item conditional on its quality being higher than $x^*+c$. This directly gives a bound for the average utility of the agents. The proof of Lemma~\ref{lem:nodivub} is given in Appendix~\ref{sec:omitted-proofs}.

In the other extreme case where agents are fully diverse, the society's average expected utility, $A(1,T)$, is even worse. We show that it is always the case that a fully diverse population obtains lower utility than a non-diverse population and so, in particular, its utility is also bounded by a constant independent of $T$.

\begin{lemma}
\label{lem:nodiv}
For any $T$, $A(1,T) \leq A(0,T)$. 
\end{lemma}

\noindent To prove Lemma~\ref{lem:nodiv}, we observe that the first agent in any society performs the same search no matter the diversity level and so achieves the same expected utility. When the diversity level is 0, the expected utility improves across rounds, while when the diversity level is 1, the expected utility stays the same across rounds, implying the result. The proof of Lemma~\ref{lem:nodiv} is given in Appendix~\ref{sec:omitted-proofs}.

In both of these extreme cases, the average expected utility is bounded by a constant that depends on the search cost $c$ but not on the number of rounds. However, as we show, any intermediate diversity level $\sigma \in (0,1)$ can be significantly better than those extreme cases as the average expected utility in those cases grows with the number of rounds. Thus, it becomes clear that societies with intermediate diversity levels achieve higher average expected utility than those with extreme diversity levels.

Our main result is a characterization of how different diversity levels $\sigma, \sigma' \in (0,1)$ compare for different numbers of rounds $T$. We identify three regions, depending on the number of rounds, that affect which diversity level is better than the other.

\xhdr{When the number of rounds $T$ is small,} the smaller diversity level  $\sigma'$ achieves higher average expected utility. This is because the benefit of broader past explorations is not as significant compared to knowing the values of a smaller subset of items more accurately. We show numerically in Section~\ref{sec:simulations} that this is indeed the case and that initially, for very few rounds, almost no diversity is preferable.

\xhdr{When $T$ takes intermediate values,} the value of exploring a larger number of options becomes more significant than higher accuracy on fewer options and thus the higher diversity level $\sigma$ is preferred. Theorem~\ref{thm:main-intermediate} shows the existence of this region under a technical assumption which requires a small gap between $\sigma'$ and $\sigma$.

\begin{theorem}[Intermediate $T$]\label{thm:main-intermediate}
  Let $\sigma, \sigma' \in (0,1)$ be two diversity levels with $\sigma' < \sigma \cdot e^{ - O(\frac1{1-\sigma^2}) }$. There exists an interval $(T_1,T_2)$ such that $A(\sigma,T) > A(\sigma',T)$ for all $T \in (T_1,T_2)$, i.e. the large diversity level $\sigma$ outperforms the small diversity level $\sigma'$. Here $T_1$ and $T_2$ depend on $\sigma$ and $\sigma'$.
\end{theorem}

\xhdr{When $T$ is very large,} the smaller diversity $\sigma'$ level becomes again preferable to the higher level $\sigma$ as in both cases a comparably large number of options have been explored in which case it is again more beneficial to know the value of the explored items more accurately. This region is described in Theorem~\ref{thm:main-large}.

\begin{theorem}[Very large values of $T$]\label{thm:main-large}
  Let $\sigma, \sigma' \in (0,1)$ be two diversity levels with $\sigma' < \sigma$. There exists an interval $(T_3,\infty)$ such that $A(\sigma',T) > A(\sigma,T)$ for all $T \in (T_3,\infty)$, i.e. the small diversity level $\sigma'$ outperforms the large diversity level $\sigma$. Here $T_3$ depends on $\sigma$ and $\sigma'$.
\end{theorem}

Theorems~\ref{thm:main-intermediate} and~\ref{thm:main-large} show that the diversity level $\sigma$ is preferable to a small diversity level $\sigma'$ for intermediate number of rounds but as the number of rounds increases the smaller diversity level becomes more socially beneficial. The proofs of the theorems are presented in Section~\ref{sec:large}.

We note that even though $A(\sigma',T)$ is guaranteed to be larger than $A(\sigma,T)$ for large values of $T$ in the third region, this will happen at extremely large values of $T$. 
In fact, in all simulations presented in the next sub-section, we ran up to $10^6$ rounds for different values of $c$, we did not observe any transition from the second region where $A(\sigma',T) < A(\sigma,T)$ to the third region where $A(\sigma',T) > A(\sigma,T)$. This is because the behavior of the $A(\sigma,T)$ is tightly characterized by the upper-bound we will show in Proposition \ref{prop:inter}, i.e.,
$\sqrt{1-\sigma^2}  \cdot \sqrt{2 \ln(1+\sigma^2\ln( T))}$, which shows that the dependence on $T$ is doubly-logarithmic and thus extremely large values of $T$ are needed to reach the third region. Thus, for all practical purposes, we observe in simulations that the optimal diversity level increases monotonically with the number of rounds.

\input{simulation}

\subsection{Behavior for intermediate and large values of $T$}\label{sec:large}

We now move on to prove Theorems~\ref{thm:main-intermediate} and~\ref{thm:main-large}. Our main technical contribution that drives the proof of those Theorems is a sharp characterization of the average expected utility for any fixed diversity level $\sigma\in(0,1)$. Proposition~\ref{prop:inter} gives a closed form expression for the average expected utility which is tight up to additive constants.

\begin{proposition}
\label{prop:inter}
For any $\sigma\in(0 ,1)$, and any $T > 0$, $$A(\sigma,T) = \sqrt{1-\sigma^2}  \cdot \sqrt{2 \ln(1 + \sigma^2 \ln T)} \pm O_c(1).$$ 
\end{proposition}

The proof of Proposition \ref{prop:inter} can be found in Appendix \ref{sec:omitted-proofs}. The main observation in establishing the statement is that the expected utility $A(\sigma,T)$ is well approximated by the maximum quality of the revealed items after $T$ rounds (Lemma~\ref{lem:quality}). Using this observation, we obtain the result by giving upper and lower bounds to the best item that is revealed in the end which in turn is closely related to the number of revealed items. To give the lower-bound (Lemma~\ref{lem:interlb}), we bound the rate that new items are revealed by bounding the probability that all available options are unsatisfactory for a given agent, i.e. when his true value is negative for all currently explored items. For the upper-bound (Lemma~\ref{lem:ub}), we show that the probability that a new item is explored is exponentially small with the number of high quality items already known.
Thus as we show, the number of items $n$ explored after $T$ rounds is logarithmic in $T$. We obtain the bound on the maximum quality of a revealed item by noting that the maximum of $n$ gaussian random variables with variance $1-\sigma^2$ is approximately $\sqrt{1-\sigma^2}  \cdot \sqrt{2 \ln n}$. This gives us the approximate expression of the expected utility $A(\sigma,T)$.

Given the closed form expression of Proposition~\ref{prop:inter}, we can easily obtain Theorems~\ref{thm:main-intermediate} and~\ref{thm:main-large} as corollaries.

\begin{prevproof}{Theorem}{thm:main-intermediate}

By Proposition \ref{prop:inter}, there exists a constant $C_1$ such that 
\[
\left| A(\sigma,T) - \sqrt{1-\sigma^2}  \cdot\sqrt{2 \ln(1+\sigma^2\ln( T))} \right| \leq C_1. 
\]

For notation convenience, define constant $C_2 = \sqrt{2 \ln(2)}$. Set $T_1 = \exp\left(\frac{1}{\sigma^2} \exp\left(\frac{(2C_1+C_2)^2}{2(1-\sigma^2)}\right)\right) $ and $T_2 = \exp\left(\frac{1}{\sigma'^2}\right)$.

For any $T < T_2$, we can upper bound $A(\sigma',T)$ by constant $C_1 + C_2$:
\begin{align*}
&A(\sigma',T) \leq  \sqrt{1-\sigma'^2}  \cdot\sqrt{2 \ln(1+\sigma'^2\ln( T))} + C_1  \\
< &\sqrt{2\ln\left(1+\sigma'^2 \cdot \frac{1}{\sigma'^2}\right)} + C_1 \leq C_2 + C_1.
\end{align*}

For any $T > T_1$, we can lower bound $A(\sigma, T)$ by constant $C_1+C_2$: 
\begin{align*}
A(\sigma,T) &\geq  \sqrt{1-\sigma^2}  \cdot\sqrt{2 \ln(1+\sigma^2\ln( T))} - C_1  \\
&> \sqrt{1-\sigma^2}  \cdot\sqrt{2 \ln\left(1+\sigma^2  \cdot \frac{1}{\sigma^2} \exp\left(\frac{(2C_1+C_2)^2}{2(1-\sigma^2)}\right)  \right)}  -C_1\\
&> 2C_1 + C_2 - C_1 = C_1 +C_2.
\end{align*}

Therefore for any $T \in (T_1,T_2)$, $A(\sigma,T) > A(\sigma',T)$. In order to make sure $T_1<T_2$, we need $\sigma' < \sigma \cdot \exp\left(-\frac{(2C_1+C_2)^2}{4(1-\sigma^2)}\right)$. 
\end{prevproof}

\begin{prevproof}{Theorem}{thm:main-large}

For any $\sigma' < \sigma$, we have that $\lim_{T \rightarrow \infty} \frac {A(\sigma,T)} {A(\sigma',T)} = \frac{\sqrt{1-\sigma^2}}{\sqrt{1-\sigma'^2}} < 1$.
This immediately implies that for any $\sigma' < \sigma$, there exists a large enough $T_3$ such that $A(\sigma,T) < A(\sigma',T)$ for all $T > T_3$, as required by Theorem~\ref{thm:main-large}.

\end{prevproof}

%% file: simulation.tex

\subsection{Behavior for small $T$ - Numerical Simulations}\label{sec:simulations}
\begin{figure}[ht]
\centering
\begin{subfigure}{.50\textwidth}
  \centering
  \includegraphics[width=0.9\linewidth]{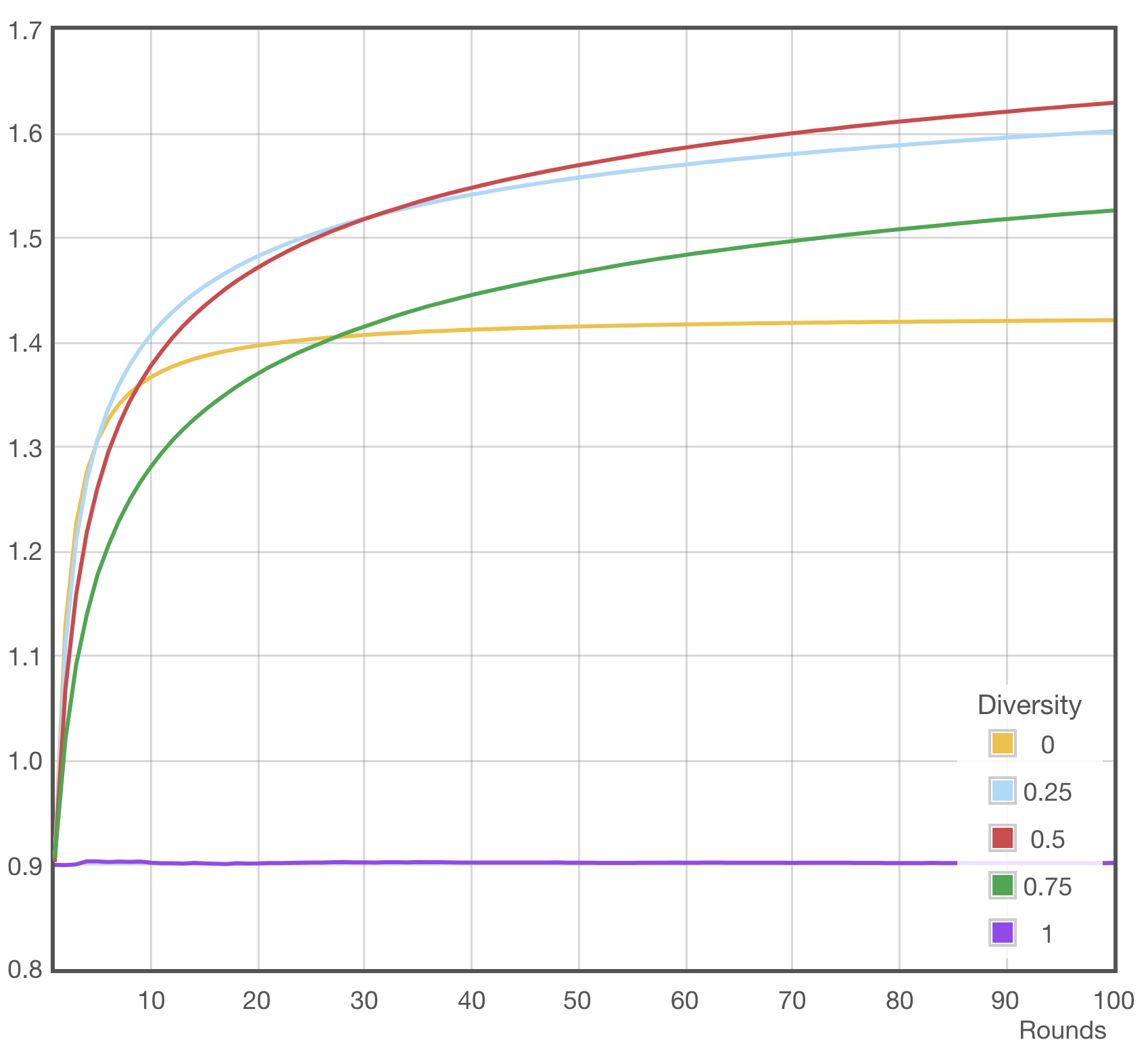}
  \caption{$c=0.1$}
  \label{fig:sub1}
\end{subfigure}%
\begin{subfigure}{.50\textwidth}
  \centering
  \includegraphics[width=0.9\linewidth]{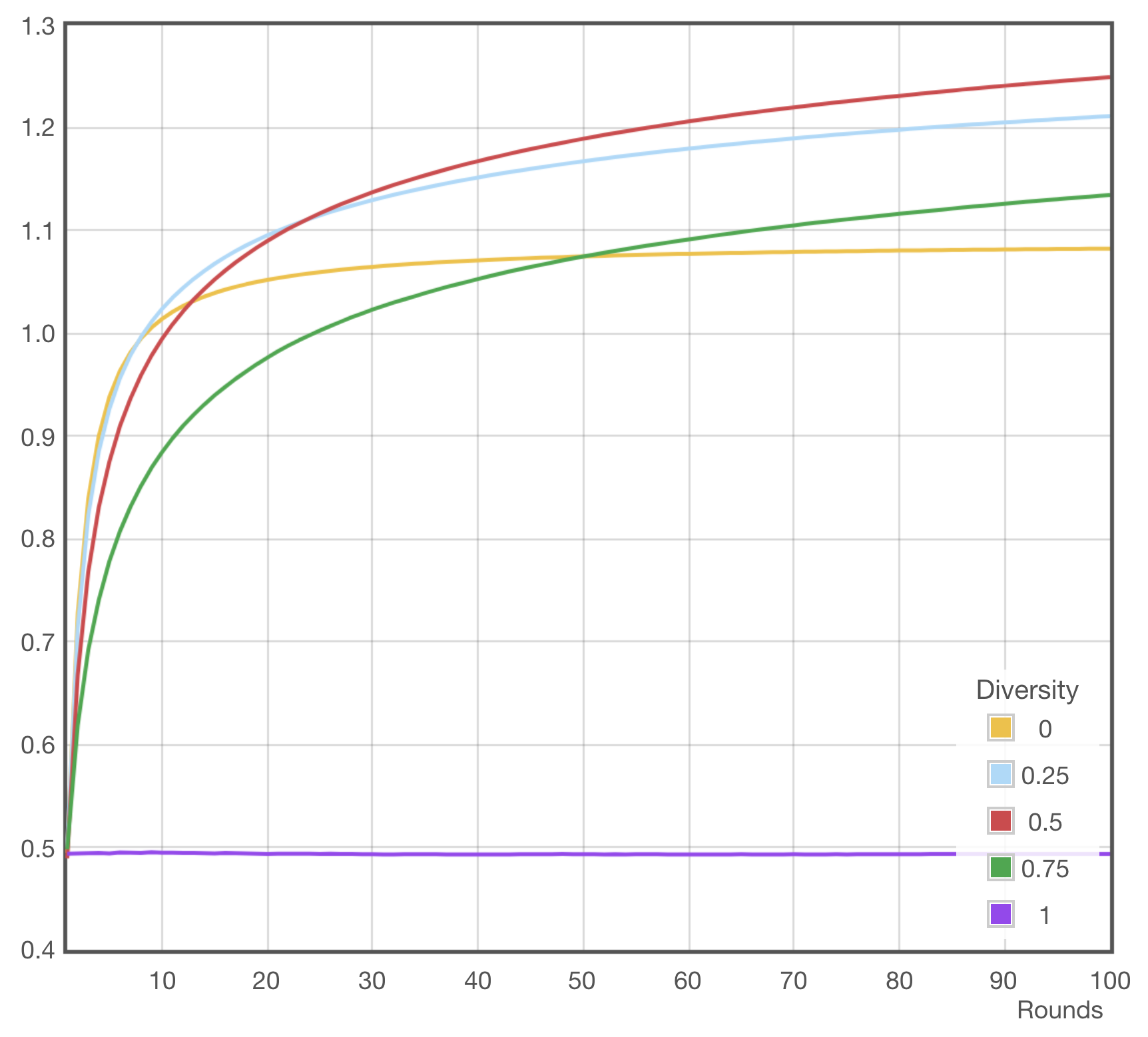}
  \caption{$c=0.2$}
  \label{fig:sub2}
\end{subfigure}
\caption{Simulations of the average expected utility for different number of rounds and diversity levels. 
}
\label{fig:simulation}
\end{figure}

Before we proceed with the proof of Theorems~\ref{thm:main-intermediate} and \ref{thm:main-large}, we present some simulations that show the evolution of society's expected average expected utility with the number of rounds $T$ in the ``revealed quality'' model. Figure~\ref{fig:simulation}, shows the plots $A(\sigma,T)$ for different number of rounds $T$ up to 100 and different diversity levels $\sigma \in \{0, 0.25, 0.5, 0.75, 1\}$. We present two plots, one for $c=0.1$ and one for $c=0.2$. For each of the different scenarios, we run $10^5$ independent runs.

As can be seen from the figure, all different diversity levels start at the same point for $T=1$. The full diverse case $\sigma=1$ remains constant at that value achieving the lowest possible utility than all diversity levels as every agent searches from scratch every time without reusing any information from the past. The case without diversity ($\sigma=0$) on the other hand, initially has an increasing average social utility but as we've shown in Lemma~\ref{lem:nodivub}, approaches a constant value after few rounds.

Diversity levels $\sigma = 0.25, 0.5,$ and $0.75$, achieve lower expected utility than the no-diversity case in the first rounds but eventually all outperform it  as described by Theorem~\ref{thm:main-intermediate}. The lowest diversity level $\sigma=0.25$ manages to do this faster but also gets outperformed by the higher diversity level $\sigma=0.5$ in later rounds. The diversity level $0.75$ also outperforms $0.25$ after $\sim 10^3$ rounds (not shown in the graph) but doesn't seem to outperform $0.5$ even after simulating for $10^6$ rounds. This is an indication that the gap in diversity levels required by Theorem~\ref{thm:main-intermediate} is necessary for statement to hold.

%% file: diamond.tex

\section{An Extension with Rare but High Reward Items}\label{sec:diamond}

We extend our results to a setting where agents' values are correlated in a more structured way. In this extension, we introduce an additional term of objective value in the quality of an item, i.e. we write $q_i$ as $d_i + r_i$ where $d_i$ takes value $D$ with probability $p = \Theta(D^{-1})$ and 0 otherwise while $r_i$ is sampled from $\N(0,1-\sigma^2)$ as before. We call an item with $d_i = D$ a ``diamond'' due to its rarity and high value. We assume that $c \ge p\cdot D$ so that the exploration cost is larger than the expected value of an item.

In this diamond model, diversity gives significant help in early rounds. Our main result in this model is that the existence of a diamond can significantly change the behavior and average utilities of agents of different diversity levels.

\begin{theorem}\label{thm:diamond}
  Consider any constant diversity level $\sigma < 1$. If the number of rounds $T$ satisfies $\omega(D) \leq T \leq  \exp(o(D))$, then $A(\sigma,T) = o(1) \cdot A(1,T)$ as the diamond value $D$ increases.
\end{theorem}

In contrast to the previous section, where exploration in the full diverse setting is meaningless as future agents cannot benefit from the experience of their predecessors, Theorem~\ref{thm:diamond} shows that if there is even a tiny probability that a very desirable and socially beneficial alternative is identified, diversity helps significantly. Therefore, we show that in societies where even fully diverse agents can agree on a clearly superior outcome which might occur with very small probability, larger diversity levels work in society's favor as they help identify this option faster.

The full proof of Theorem \ref{thm:diamond} can be found in Appendix \ref{sec:omitted-proofs}. Here we show some important steps of the proof.

To prove Theorem~\ref{thm:diamond}, we first show that when the diversity level is 1 and $T = \omega(1/p)$, the expected average utility is close to the diamond value $D$. 

\begin{lemma}
\label{lem:diamonddiv}
$A(1,T)$ is at least $(1-o(1)) D$ for $T = \omega(1/p)$. 
\end{lemma}

The main idea behind the proof of Lemma \ref{lem:diamonddiv} is to show that at least one new item is searched in each round before a diamond is discovered and thus in $T = \omega(1/p)$, the diamond will be found. The proof of Lemma \ref{lem:diamonddiv} can be found in Appendix \ref{sec:omitted-proofs}.

Next we proceed to bound the expected average utility for any diversity level $\sigma < 1$, when $ T \leq  \exp(o(D))$, $A(\sigma,T)$. 
We start with the extreme case when $\sigma = 0$. The proof is similar to the proof of Lemma \ref{lem:nodivub} and is given in Appendix~\ref{sec:omitted-proofs}.
\begin{lemma}
\label{lem:0divdm}
$A(0,T)$ is at most $D \cdot p +O(1)$.
\end{lemma}

The case for intermediate diversity level $\sigma \in (0,1)$ is more complicated. We give an upper bound on the average utility in the appendix (Lemma \ref{lem:dmub}).

%% file: gaussian.tex

\section{Tools for Gaussian Distributions}

In this section, we list several facts and lemmas about Gaussian distributions. Throughout this document, we use $\Phi$ to refer to the $CDF$ of the Gaussian $\N(0,1)$ and $\varphi$ to refer to the density function of the Gassian $\N(0,1)$.

\begin{fact}
\label{fact:gaussiancdf}
Let $\Phi(t)$ be the CDF of standard Gaussian distribution then for $t> 0$,
\[
\frac{1}{\sqrt{2\pi}}\cdot \frac{t}{t^2+1}\cdot\exp(-t^2/2)\le 1-\Phi(t) \le \min\{ \frac{1}{\sqrt{2\pi} t}, 1\} \exp(-t^2/2).
\]
\end{fact}

\begin{fact}[\cite{Baricz08}]
\label{fact:gaussianhr}
Let $\Phi(t)$ be the CDF of standard Gaussian distribution. Let $\varphi(t)$ be the density of standard Gaussian distribution. For any $t >0 $, 
\[
\frac{\varphi(t)}{1-\Phi(t)} \leq \frac{t}{2} + \frac{\sqrt{t^2+4}}{2}. 
\]
\end{fact}
The following two lemmas can be proved by standard techniques about Gaussian distributions. We omit their proofs because of space limit.
\begin{lemma}
\label{lem:gaussian}
For any $T > 0$, let $x_1, ...,x_T$ be sampled from $\N(0,\sigma^2)$ independently. Then we have
\[
Pr\left[ \min_{i \in [T]} \frac{x_1 + \cdots + x_i}{i} \leq -2\sigma \right] \leq 1/4. 
\]
\end{lemma}
\begin{lemma}
\label{lem:gaussianmax}
For any $x^*, T > 0$, let $x_1, ...,x_T$ be sampled from $\N(0,\sigma^2)|_{\ge x^*}$ independently. Then we have
\[
\E\left[ \max_{i \in [T]} x_i \right] \leq \sigma \sqrt{ 2 \ln (2T) } + x^*. 
\]
\end{lemma}

%% file: omitted-proofs.tex

\section{Omitted Proofs }
\label{sec:omitted-proofs}

\subsection{Ommitted Proofs of Lemma \ref{lem:nodivub} and Lemma \ref{lem:nodiv}}
\begin{lemma}[Restatement of Lemma \ref{lem:nodivub}]
\label{lem:nodivubapp}
For any $T$, $A(0,T)\leq\sqrt{ \ln \left( \frac{1}{2\pi c^2}\right)}  + 2$. 
\end{lemma}

\begin{proof}
The proof plan is to first bound by a constant the Weitzman index $x^*$ of an unsearched item, i.e.\ one whose quality is not yet revealed. We then argue that once some item with sufficient quality has been identified (with quality $q_i > x^*+c$), no other item will be explored in future rounds. We conclude the proof by upper-bounding the expected quality of such an item conditional on its quality being higher than $x^*+c$. This directly gives a bound for the average utility of the agents. 

We start by bounding the Weitzman index $x^*$ of an unsearched item. For an unsearched item, the distribution of the value is $\N(0,1)$. The Weitzman index $x^*$ satisfies $\E[\max(X -x^*,0)] = c$ where $X \sim \N(0,1)$. Using Fact \ref{fact:gaussiancdf} on Gaussian distributions, we see
\begin{eqnarray*}
c &=& \E[\max(X -x^*,0)] \\
&=& \E[X -x^* | X-x^* >0] \cdot Pr[X-x^* >0]\\
 &=& \left(-x^* + \frac{\varphi(x^* )}{\Phi(-x^*)}\right) \cdot \Phi(- x^*)\\
 &=& (-x^*) (1- \Phi( x^*)) + \varphi(x^*) \\
 &\leq& (-x^*) \left(\frac{\varphi(x^*)}{x^*} -\frac{\varphi(x^*)}{(x^*)^3} \right) + \varphi(x^*)\\
&=& \frac{\varphi(x^*)}{(x^*)^2} = \frac{1}{\sqrt{2\pi} \cdot (x^*)^2 \cdot \exp( (x^*)^2/2)}.
\end{eqnarray*}
This implies an upper bound on the Weitzman index of an unsearched item: 
\[
x^* \leq \max \left(1, \sqrt{ \ln \left( \frac{1}{2\pi c^2}\right)}\right).
\]
 
When an item has been searched with observed quality $q_i$, since the diversity level is 0, the values of this item for future agents are determinstically $q_i$. Therefore the Weitzman index of this item is simply $q_i -c$. 

Once some item has quality $q_i > c + x^*$, its Weitzman index is larger than $x^*$ which means it will be searched before any unsearched items. And once this item is searched in a round, we know its realized value is $q_i$ which is larger than $x^*$. Therefore no more unsearched items will be searched.

Let $Q$ be the quality of the first item with quality larger than $c + x^*$. We have $\E[Q] = \E_{X \sim \N(0,1)}[X|X > c+x^*]$. We know that starting from the round this item is searched, each agent's utility is at most $Q - c$ and before this item is ever searched, each agent's utility is at most $c+x^* - c \leq Q - c$. Therefore we can bound the average utility as following: 
\begin{eqnarray*}
A(0,T) &\leq& \E[Q - c] \leq \E_{X \sim \N(0,1)}[X|X > c+x^*] - c \\ 
&=& \frac{\varphi(c+x^*)}{1- \Phi(c+x^*)} - c \\
&\leq& \frac{x^* - c + \sqrt{(x^*+c)^2 + 4}}{2} \text{    ( by Fact \ref{fact:gaussianhr})}\\
&\leq& x^*+1 \\
&\leq& \sqrt{ \ln \left( \frac{1}{2\pi c^2}\right)}  + 2. 
\end{eqnarray*}
\end{proof}

\begin{lemma}[Restatement of Lemma \ref{lem:nodiv}]
\label{lem:nodivapp}
For any $T$, $A(1,T) \leq A(0,T)$. 
\end{lemma}

\begin{proof}
First notice that when the diversity level is 1, values are completely independent of the history and have the same distributions across agents. Therefore, the expected utilities of all agents are the same. Therefore $A(1,T)$ equals the expected utility of the first agent. 

When the diversity level is 0, consider the following strategy for agents except the first agent: just search the item chosen by the first agent and choose it. By using this strategy, we know all agents will get at least the utility of the first agent. By the optimality of Weitzman index strategy, we know $A(0,T)$ will be at least the expected utility of the first agent. 

Finally, notice that the expected utility of the first agent does not depend on the diversity level. This implies $A(0,T) \geq A(1,T)$. 
\end{proof}

\subsection{Omitted Proof of Proposition~\ref{prop:inter}}

We proceed to show Proposition~\ref{prop:inter}. We first relate the average expected utility $A(\sigma,T)$ to the expectation of the positive part of the highest quality explored during the first $T$ rounds when the diversity level is $\sigma$, denoted by $M(\sigma,T)$. Here the positive part of $x$ means $\max(x,0)$. In Lemma \ref{lem:quality}, we show that $A(\sigma,T)$ is sandwiched between $M(\sigma,\gamma T)$ and $M(\sigma,T)$. Subsequently, we provide sharp bounds on $M(\sigma,T)$ and show that $M(\sigma,\gamma T)$ and $M(\sigma,T)$ are very close for large $T$.

\begin{lemma}
\label{lem:quality}
For any diversity level $\sigma\in(0,1)$ and any $\gamma \in (0,1)$, $$\left(1-\gamma\right)( M\left(\sigma,\gamma T\right) - c) \leq A(\sigma,T) \leq M(\sigma,T) + O_c(1).$$
\end{lemma}

\begin{proof}
Let's first prove the left hand side. For each agent after round $\gamma T$, a sub-optimal strategy is to explore and pick the item with highest quality searched before the end of round $\gamma T$ and leave. This gives utility $M(\sigma,\gamma T) - c$ in expectation. By the optimality of the Weitzman index strategy, agents after round $\gamma T$ should at least get $M(\sigma,\gamma T) - c$ in expectation. Again by the optimality of the Weitzman index strategy, agents in the first $\gamma T$ rounds should get utility at least 0 in expectation. Thus, we have $A(\sigma,T) \geq \left(1-\gamma\right) (M(\sigma,\gamma T) - c)$. 

Now let's consider the right hand side. The main idea is to upper bound the expected utility of the last round $T$ by $M(\sigma,T) + O_c(1)$. 

Let $n$ be the number of items explored in the first $T$ rounds, with qualities $q_1,...,q_n$. In round $T$, the agent's expected utility will be at least as large as the expected utility of agents in previous rounds as more information is revealed. Let $j$ be the item where the $T$-th agent stops his search. We fix $j, n$ and $q_1,..,q_n$ in this argument and expectations are taken over randomness of subjective scores of the $T$-th agent. $T$-th agent's value for any previously explored option is less than $x^*_j$, the Weitzman index of item $j$, as he chose to explore item $j$. Moreover, since he decided to stop at $j$, his value at $j$, $q_j + s_{j,T}$ must be higher than $x^*_{j+1}$, the Weitzman index of the next item to search in round $T$. After the search, the $T$-th agent might take item $j$ or some other items with value smaller than $x^*_j$. Thus, his expected utility, is at most $\max(q_j + s_{j,T},x^*_j)- c$ and $s_{j,T}$ is sampled from $ \N(0,\sigma^2)$ conditioned on $q_j + s_{j,T} > x^*_{j+1}$. 
 
Using Lemma \ref{lem:gaussianmax}, we can upper bound the expected utility of the agent in round $T$ by 
\begin{align*}
\E[\max(q_j + s_{j,T},x^*_j)] &\leq \E_{X\sim \N(0,\sigma^2)}[q_j + X|q_j +X > x_{j}^*] \\
&\leq q_j + \E_{X\sim \N(0,\sigma^2)}[X|X > \max(x_{j}^* -q_j,0)] \\
&\leq q_j + \max(x_{j}^* -q_j,0) + O(1) ~~~~~~~~\text{(by Lemma \ref{lem:gaussianmax})}
\end{align*}
 
Recall $x^*$ is the Weitzman index of an unexplored item. If item $j$ has been searched before round $T$, we have $x^*_j =x^*$. We can simply get
\[
\E[\max(q_j + s_{j,T},x^*_j)] \leq \max(q_j,0) + x^* +O(1).
\]

If item $j$ hasn't been searched before round $T$, we bound the difference between $x_j^*$ and $q_j$. The definition of Weitzman index gives us,
\begin{align*}
c &= \E_{X \sim \N(0,\sigma^2)} [\max(X - ( x_j^*-q_j),0)] \\
&\leq \E_{X \sim \N(0,1)} [\max(X - ( x_j^*-q_j),0)]. 
\end{align*}
Note that the last expression is strictly less than $c$ when $x_j^* - q_j > x^*$ implying that $x_j^* - q_j \le x^*$. Then we get 
\begin{align*}
\E[\max(q_j + s_{j,T},x^*_j)] &\leq q_j + x^* + O(1)
\end{align*} 
To sum up, no matter item $j$ has been searched or not before round $T$, the expected utility of the agent in round $T$ is at most $M(\sigma,T) +  x^* + 1$. Since the agent in round $T$ has expected utility at least as high as expected utilities of agents in previous rounds, we conclude that $A(\sigma,T) \leq M(\sigma,T) +  x^* + O(1)$.
\end{proof}

In Lemma \ref{lem:interlb}, we give a lower bound on $M(\sigma,T)$. The main idea of the proof is to give a lower bound on the number of searched items and then show that $M(\sigma,T)$ is roughly the expected value of the maximum of that many independent Gaussians. The lower bound on the number of searched items can be achieved by considering events (indicated as $W_t$ in the proof) when all the previously searched items have low values (caused by very negative subjective scores). 

\begin{lemma}
\label{lem:interlb}
For any diversity level $\sigma\in(0,1)$, and any $T>\exp(13/\sigma^2)$, $$M(\sigma,T) \geq  \sqrt{1-\sigma^2}  \cdot \sqrt{2  \ln( \sigma^2 \ln T)} - O(1).$$ 
\end{lemma}

\begin{proof}
  To prove the statement, we fix a sufficiently large parameter $n$ and we show that with high probability after a sufficiently large number of rounds $T$, two events will happen:
  \begin{itemize}
    \item The maximum quality among the first $n$ items is at least $h = \sqrt{1-\sigma^2} \cdot \sqrt{2} \cdot (\sqrt{\ln n } - 1)$,
    \item At least one agent will explore an item with quality higher than $h$.
  \end{itemize}

We first bound the probability that the maximum quality of the first $n$ items is less than $h$. As the qualities of all items are independent, this is a standard calculation about the maximum of independent Gaussian random variables. Setting $g = \frac{h}{\sqrt{1- \sigma^2}}$, we have:
	\begin{align*}
		&~~~~Pr\left[ \max(q_1,...,q_{n}) \leq h \right]\\ &= \left(Pr\left[q_1 < h\right] \right)^{n}  
		=  \Phi \left( \frac{h}{\sqrt{1- \sigma^2}} \right)^{n} \\
    &\leq \left(1- \frac{\exp\left(- \frac{1}{2}\cdot g^2\right)}{\sqrt{2\pi} }\cdot \frac{g}{g^2+1}\right)^n & \text{from Fact~\ref{fact:gaussiancdf}}\\
		&\leq \exp\left(-n \cdot \frac{\exp\left(- \frac{1}{2}\cdot g^2\right)}{\sqrt{2\pi} }\cdot \frac{g}{g^2+1}\right) \\
		&\leq \exp\left(- \frac {\exp{(2\sqrt{\ln n})}} {  e \sqrt{\pi}}\cdot \frac{\sqrt{\ln n}-1}{2(\sqrt{\ln n}-1)^2+1}\right) \leq \frac{ 1 } { \sqrt{\ln n} }   & \text{for $n \ge 5$}.
	\end{align*}

Assuming that there is some item with quality higher than $h$ among the first $n$, we now bound the probability that some agent will explore it. Let $W_t$ be the indicator of the event that in round $t$, for all items $i \in [n]$ with $q_i \leq h$,  $s_{i,t} < -h$. It is clear that if $W_t = 1$ for some $t$, $v_{i,t} < 0$ for all items with quality less than $h$ so the agent $t$ will keep exploring until he finds the item with quality higher than $h$.

We lower bound the probability that $W_t = 1$. We have that: 
	\begin{align*}
		Pr[W_t = 1] &\geq \left(Pr\left[s_{i,t} < -h \right] \right)^{n} = \left(1-\Phi\left(\frac{h}{\sigma}\right)\right)^{n} \\
		&\geq \exp\left(-\frac{1}{2}\cdot \left(\frac{h}{\sigma}\right)^2\cdot n \right) \geq \exp\left(-\frac{n \ln n}{\sigma^2} \right).
	\end{align*}

Setting $n = \frac{Z}{\ln(Z)}$ for $Z = \sigma^2 \ln T$, we get that $Pr[W_t = 1] \ge \frac {\ln Z} {T}$. Thus, the probability that no agent will find the high quality item is at most: $Pr[\forall t, W_t = 0] \le \left( 1 - \frac {\ln Z} {T} \right)^T \le \frac {1} {Z} \le \frac{ 1 } { \sqrt{\ln n} }$.

Thus, overall we have that if $n \ge 5$
	\begin{align*}
M(\sigma,T) &\geq (1 - \frac{ 2 } { \sqrt{\ln n} }) h \ge \sqrt{1-\sigma^2} \cdot \sqrt{2} (\sqrt{ \ln n } - 3 ) \\
&\ge \sqrt{1-\sigma^2} \cdot \sqrt{2} (\sqrt{ \ln Z } - 4 ) &\text{for $Z > n \ge 5$}\\
&\ge \sqrt{1-\sigma^2} \cdot \sqrt{2 \ln ( \sigma^2 \ln T) } - 6.
	\end{align*}
Thus, setting $T$ so that $n \ge 5$, we get the required bound.
\end{proof}

In addition to this lower bound, in Lemma \ref{lem:ub}, we provide an upper bound on $M(\sigma,T)$. The main idea of the proof is to consider items with good qualities and show that the chance of searching a new item decreases exponentially with the number of such items. We perform the analysis using a coupling argument to an imaginary setting in which a larger number of items is always explored.

\begin{lemma}
\label{lem:ub}
For any diversity level $\sigma\in(0,1)$, $$M(\sigma,T) \leq \sqrt{1-\sigma^2}  \cdot\sqrt{2 \ln(1+\sigma^2\ln( T))}+ O_{c}(1).$$
\end{lemma}
\begin{proof}
Let $I(t)$ be a random variable that indicates the number of explored items after the agent in the $t$-th round has finished exploring, label items in the order in which they were first explored, and let $q_1,q_2,...,q_{I(T)}$ be random variables that denote the quality of the items explored in the first $T$ rounds.  Then, in this notation, $M(\sigma,T)$, the expected positive part of the maximum quality of the explored items at the end of round $T$, is $\E[ \max(0,\max_{i\le I(T)} q_i) ]$ .

To bound this expression, we define the following alternative exploration procedure that is significantly easier to analyze. We will show that there is always a coupling between the original and the alternative processes so that more items are always explored in the latter. Let $x^*$ be the Weitzman index that satisfies $\E[\max(X -x^*,0)] = c$ when $X \sim \N(0,1)$. In our alternative process, every agent goes through every item $i$ starting from $i=1$ and, whenever $q_i \ge x^* + c + 1$, flips a coin that has probability $\Phi(1/\sigma)$ of coming heads. If the coin comes heads he stops, otherwise he continues on to the next item $i+1$. Let $I'(t)$ be the maximum number of items that have been explored by the agents in the first $t$ rounds of this alternative procedure.
  
We define a coupling such that $I(t) \le I'(t)$. In the original process, we have that every agent $t$ has value $q_i + s_{i,t}$ for item $i$. In the alternative process, if $q_i \ge x^* +c+1$, we couple the coin with the event that $s_{i,t} \ge -1$ which also has probability $\Phi(1/\sigma)$. 
Assuming inductively that $I(t-1) \le I'(t-1)$, we will show that $I(t) \le I'(t)$. This follows immediately if $I(t) \le I'(t-1)$ since we have  that $I'(t-1) \le I'(t)$.
Otherwise, when $I(t) > I'(t-1)$, the original agent in the $t$'th round must have explored new item(s). Therefore, he must have $s_{i,t} \le -1$ for all items $i < I(t)$ for which $q_i \ge x^* + c + 1$ as he must have explored all of those first before reaching item $I(t)$. By definition of the coupling, the alternative agent will not stop in any item $i$ such that $i < I(t)$ and thus $I'(t) \ge I(t)$.
  
Thus, $M(\sigma,T) = \E[ \max(0,\max_{i\le I(T)} q_i) ] \le \E[\max(0, \max_{i\le I'(T)} q_i) ]$.  We complete the proof by bounding $\E[\max(0, \max_{i\le I'(T)} q_i) ]$ in Lemma~\ref{lem:boundM}.
\end{proof}

\begin{lemma}\label{lem:boundM}
	Let $I'(T)$ be as defined in the proof of Lemma~\ref{lem:ub}.  Then $$\E[\max(0, \max_{i\le I'(T)} q_i) ]\leq \sqrt{1-\sigma^2} \cdot  \sqrt{2 \ln (1+\sigma^2 \ln T )} + O_c(1).$$
\end{lemma}

\begin{proof}
For notation convenience, define $p_k = \Pr[ |\{ i | i \le I'(T) \wedge q_i \ge x^* + c + 1\}| = k ] $. 
	Letting $m_k$ denote the expected maximum of $k$ random variables distributed according to $\N(0,1-\sigma^2)|_{\ge x^* + c + 1}$, we have that
$
	\E[ \max(0,\max_{i\le I'(T)} q_i) ]\le \sum_k m_k p_k
	$.	From Lemma~\ref{lem:gaussianmax}, we have that $m_k \le \sqrt{1-\sigma^2}  \cdot \sqrt{2 \ln (2 k)} + x^* + c + 1 \le \sqrt{1-\sigma^2}  \cdot \sqrt{2 \ln k} + x^* + c + 3$ and thus we get that
	\begin{align*}
	&\E[\max(0, \max_{i\le I'(T)} q_i) ]\\
 \le& \sqrt{1-\sigma^2} \sum_k \sqrt{2 \ln k} \cdot p_k + (x^*+c+3)\\
 \le& \sqrt{1-\sigma^2} \cdot  \sum_{k \ge z} \sqrt{2 \ln k} \cdot p_k +\sqrt{1-\sigma^2}\cdot \sqrt{2 \ln z} + (x^*+c+3)
	\end{align*}
	for all $z \ge 1$.

  We now bound $p_k= \Pr[ |\{ i | i \le I'(T) \wedge q_i \ge x^* + c + 1\}| = k ]$, i.e. the probability that $k$ items with qualities higher than $x^* + c + 1$ are explored during the $T$ rounds.
  We note that this can happen if in at least one round an agent manages to get $k$ coin flips to land tails before stopping the exploration. Every agent achieves this with probability at most $(1-\Phi(1/\sigma))^{k} \le e^{-\frac {k} {2\sigma^2}}$ and thus for all $T$ rounds, the probability that $k$ items with value higher than $x^*+c+1$ are explored is at most $T \cdot e^{-\frac {k} {2\sigma^2}}$ by a union bound. We thus get that:
	\begin{align*}
&\E[ \max(0,\max_{i\le I'(T)} q_i) ]\\
\le& \sqrt{1-\sigma^2} \cdot \left( \sqrt{2 \ln z} + T \sum_{k \ge z} e^{-\frac {k} {2\sigma^2}} \sqrt{2 \ln k} \right)
	+ (x^*+c+3) \end{align*}
	Setting $z = 7 (1 + \sigma^2 \ln T)$ we have that 
	\begin{align*}
	&T \sum_{k \ge z} e^{-\frac {k} {2\sigma^2}} \sqrt{2 \ln k} 
	\le e^{-7} e^{\frac {z} {7\sigma^2}} \sum_{k \ge z} e^{-\frac {k} {2\sigma^2}} \sqrt{2 \ln k}\\
	\le &e^{-7} e^{\frac {z} {7}} \sum_{k \ge z} e^{-\frac {k} {2}} \sqrt{2 \ln k}
	\le e^{-7} e^{\frac {z} {7}} e^{-\frac {z} {4}}
	\le e^{-7}
	\end{align*}
	and we obtain 
	\begin{align*}
	\E[ \max(0,\max_{i\le I'(T)} q_i) ]\le \sqrt{1-\sigma^2} \cdot  \sqrt{2 \ln (1 + \sigma^2 \ln T )} + x^* + c + 5.
	\end{align*}
\end{proof}
Finally, by combining Lemma \ref{lem:quality}, \ref{lem:interlb} and \ref{lem:ub}, we prove Proposition \ref{prop:inter}.

\begin{prevproof}{Proposition}{prop:inter}
 \noindent 
Let's first prove the upper bound on $A(\sigma,T)$. By Lemma \ref{lem:quality} and Lemma \ref{lem:ub}, we have
\begin{align*}
A(\sigma,T) &\leq M(\sigma,T) + O_{c}(1)\\
&\leq \sqrt{1-\sigma^2}  \cdot\sqrt{2 \ln(1+\sigma^2\ln( T))}+ O_{c}(1)\\
\end{align*}

Now let's prove the lower bound on $A(\sigma,T)$. Notice that if $T \leq \exp(13/\sigma^2)$, 
\[
\sqrt{1-\sigma^2} \cdot\sqrt{2 \ln(1+\sigma^2\ln( T))} \leq \sqrt{2 \ln(1+13)} =  O(1),
\] and therefore the lower bound is trivial. We only need to prove the lower bound when $T > \exp(13/\sigma^2)$. 

 By Lemma \ref{lem:quality} and Lemma \ref{lem:interlb}, we have
\begin{align*}
A(\sigma,T) &\geq \left(1-\frac{1}{\sqrt{T}}\right) (M(\sigma,\sqrt{T}) - c) \\
 &\geq \left(1-\frac{1}{\sqrt{T}}\right) \left(  \sqrt{1-\sigma^2}  \cdot \sqrt{2  \ln\left( \sigma^2 \ln \left(\sqrt{T}\right) \right)} - O(1) -c\right)\\
&\geq  \sqrt{1-\sigma^2}  \cdot\sqrt{2 \ln(1+\sigma^2\ln( T)) - 2} - O_{c}(1) \\
&~~~- \frac{\sqrt{1-\sigma^2}  \cdot\sqrt{2 \ln(1+\sigma^2\ln( T))} }{\sqrt{T} }\\
&\geq  \sqrt{1-\sigma^2}  \cdot\sqrt{2 \ln(1+\sigma^2\ln( T))} - O_{c}(1).
\end{align*}

\end{prevproof}

\subsection{Omitted Proofs of Section \ref{sec:diamond}}
\begin{lemma}[Restatement of Lemma \ref{lem:diamonddiv}]
\label{lem:diamonddivapp}
$A(1,T)$ is at least $(1-o(1)) D$ for $T = \omega(1/p)$. 
\end{lemma}

\begin{proof}
When diversity level is $\sigma = 1$, we have $q_i = d_i$. So $q_i$'s realized value is either 0 or $D$. For each explored item, if it has quality $0$, its Weitzman index will be lower than the Weitzman index of an unexplored item. This means that before an item with quality $D$ is explored, each agent will explore at least one new item. Therefore, at round $t$, the probability of having seen a diamond is at least $1-(1-p)^t$. Once a diamond has been found, each successive agent will get utility at least $D-c$.  Therefore, for any $t$,
$A(1,T)\geq (1-(1-p)^t)(1-t/T)(D-c).$

We want to pick $t = \omega(1/p)$ and $t/T = o(1)$. Choosing $t=\sqrt{T \cdot (1/p)}$ suffices. We get 
$
A(1,T) \geq (1-o(1))(1 - o(1)) (D-c) = (1-o(1))D.
$
\end{proof}

\begin{lemma}[Restatement  of Lemma \ref{lem:0divdm}]
\label{lem:0divdmapp}
$A(0,T)$ is at most $D \cdot p +O(1)$ in the diamond model.
\end{lemma}

\begin{proof}
The proof plan is similar to the proof of Lemma \ref{lem:nodivub}. We first bound by a constant the Weitzman index $x^*$ of an unsearched item, i.e.\ one whose quality is not yet revealed. We then argue that once some item with sufficient quality has been identified (with quality $q_i > x^*+c$), no other item will be explored in future rounds. We conclude the proof by upper-bounding the expected quality of such an item conditional on its quality being higher than $x^*+c$. This directly gives a bound for the average utility. 

We start by the Weitzman index of an unsearched item. Let $x^*$ be its Weitzman index. Let $X \sim \N(0,1)$, by the defintion of Weitzman index,
$
c =  p \cdot  \E[\max(X + D -x^*,0)] + (1-p) \cdot \E[\max(X  -x^*,0)].
$

We give upper bounds on $  \E[\max(X + D -x^*,0)]$ and $ \E[\max(X  -x^*,0)]$ separately.
For $\E[\max(X + D -x^*,0)]$, since $x^* > 0$, we have
\begin{align*}
&\E[\max(X + D -x^*,0)]\leq \E[ \max(X+D,0)] \\
=& \E[X+D] + (\E[-X|-X \geq D] - D) \cdot \Pr[-X \geq D] \\
=& D + \left(\frac{\varphi(D)}{1-\Phi(D)} - D\right) (1- \Phi(D))\\
\leq& D + \left(\frac{D}{2} + \frac{\sqrt{D^2+4}}{2} - D \right)\cdot \frac{1}{\sqrt{2\pi}\exp(D^2/2)D}\\
\leq &D +\frac{1}{\sqrt{2\pi}\exp(D^2/2)D^2}.
\end{align*}

For $\E[\max(X  -x^*,0)]$, we have
\begin{eqnarray*}
&&\E[\max(X -x^*,0)] = \E[X -x^* | X-x^* >0] \cdot Pr[X-x^* >0]\\
 &=& \left(-x^* + \frac{\varphi(x^* )}{\Phi(-x^*)}\right) \cdot \Phi(- x^*) = (-x^*) (1- \Phi( x^*)) + \varphi(x^*) \\
 &\leq& (-x^*) \left(\frac{\varphi(x^*)}{x^*} -\frac{\varphi(x^*)}{(x^*)^3} \right) + \varphi(x^*)= \frac{\varphi(x^*)}{(x^*)^2} \\
 &=& \frac{1}{\sqrt{2\pi} \cdot (x^*)^2 \cdot \exp( (x^*)^2/2)}.
\end{eqnarray*}
To sum up we get
$
c \leq p\cdot D +\frac{p}{\sqrt{2\pi}\exp(D^2/2)D^2} +\frac{1-p}{\sqrt{2\pi} \cdot (x^*)^2 \cdot \exp( (x^*)^2/2)}.
$
For notation convenience, define
$
\alpha = \frac{c - p \cdot D - \frac{p}{\sqrt{2\pi}\exp(D^2/2)D^2}}{1-p}.
$ 
We have
$
\sqrt{2\pi} \cdot (x^*)^2 \cdot \exp( (x^*)^2/2) \leq \frac{1}{\alpha}.
$
As $c - p \cdot D = \Omega(1)$ and $D = \omega(1)$, we have $\alpha = \Omega(1)$. 
This implies an upper bound on the Weitzman index of an unsearched item: 
$
x^* \leq \max \left(1, \sqrt{ \ln \left( \frac{1}{2\pi \alpha^2}\right)}\right) = O(1).
$

When an item has been searched with observed quality $q_i$, since the diversity level is 0, the values of this item for future agents are determinstically $q_i$. Therefore the Weitzman index of this item is simply $q_i -c$. 

Once some item has quality $q_i > c + x^*$, its Weitzman index is larger than $x^*$ which means it will be searched before any unsearched items. And once this item is searched in a round, we know its realized value is $q_i$ which is larger than $x^*$. Therefore no more unsearched items will be searched.

Let $Q$ be the quality of the first item with quality larger than $ c + x^*$. We have $\E[Q] = \E[q_i|q_i > c+x^*]$. We know that starting from the round this item is searched, each agent's utility is at most $Q - c$ and before this item is ever searched, each agent's utility is at most $c+x^* - c \leq Q - c$. Therefore we can bound the average utility as following: 
\begin{eqnarray*}
A(0,T) &\leq& \E[Q-c] \leq \E [q_i|q_i > c+x^*] - c \\ 
&\leq& D \cdot p+ \frac{\varphi(c+x^*)}{1- \Phi(c+x^*)} - c \\
&\leq&D \cdot p+  \frac{x^* - c + \sqrt{(x^*+c)^2 + 4}}{2}\\
&\leq& D \cdot p+x^*+1 \text{~~~~~~~~~~~(by Fact \ref{fact:gaussianhr})}\\
&=& D \cdot p+O(1).
\end{eqnarray*}
\end{proof}

We proceed to show Lemma \ref{lem:dmub} which is the main technical part behind the proof of Theorem~\ref{thm:diamond} and gives a bound on the expected maximum quality of the items explored in $T$ rounds, i.e. $M(\sigma,T)$. This implies an upper bound on the expected average utility $A(\sigma,T)$ using an identical argument to Lemma~\ref{lem:quality}. We omit the proof because of space limit. It is very similar to the proof of Lemma \ref{lem:ub}.

\begin{lemma}
\label{lem:dmub}
For any diversity level $0<\sigma<1$, in the diamond model, 
\begin{align*}
&A(\sigma, T) \leq M(\sigma,T) \\
\leq& \sqrt{1-\sigma^2}  \cdot\sqrt{2 \ln(1+\sigma^2\ln( T))}+ \frac{pD (2\sigma^2 \ln(D) + 2\sigma^2 \ln(T))}{1-\Phi\left(\frac{x^*+c+1}{\sqrt{1-\sigma^2}}\right)} + O_{c}(1) .
\end{align*}
where $x^*$ is the Weitzman index of an unsearched item. 
\end{lemma}

\begin{prevproof}{Theorem}{thm:diamond}
Given Lemma~\ref{lem:dmub}, we can now complete the proof of Theorem~\ref{thm:diamond}. We know that $x^*$ (the Weitzman index of an unsearched item) is constant as shown in the proof of Lemma~\ref{lem:0divdm}. So for a constant $\sigma < 1$, we know that $\frac{x^*+c+1}{\sqrt{1-\sigma^2}}$ is also a constant and therefore $\Phi\left(\frac{x^*+c+1}{\sqrt{1-\sigma^2}}\right)$ is bounded away from 1. Plugging $T = \exp(o(D))$ into Lemma \ref{lem:dmub}, we get that $A(\sigma,T) = o(D)$. As $A(1,T) = (1-o(1)) D$ by Lemma~\ref{lem:diamonddiv}, the theorem follows.
\end{prevproof}